\documentclass[11pt]{article}
\usepackage[utf8]{inputenc}
\usepackage[english]{babel}
\usepackage{amsmath,enumitem,amsthm,xy,array,dsfont}
\usepackage{amssymb,bbding,float,caption,multicol,multirow}
\advance\textwidth by 3cm

\advance\textheight by 3.5cm

\usepackage{tikz}
\usepackage{qtree}

\usepackage{epstopdf}

\theoremstyle{plain}
\newtheorem{theorem}{Theorem}[section]
\newtheorem{lemma}[theorem]{Lemma}
\newtheorem{proposition}[theorem]{Proposition}
\newtheorem{corollary}[theorem]{Corollary}

\theoremstyle{definition}

\newtheorem{example}[theorem]{Example}
\newtheorem{remark}[theorem]{Remark}

\numberwithin{equation}{section}

\newtheorem{Assumption}{Assumption}

\newcommand{\calF}{\mathcal{F}}
\newcommand{\calP}{\mathcal{P}}
\newcommand{\calA}{\mathcal{A}}
\newcommand{\Y}{\mathcal{Y}}
\newcommand{\calX}{\mathcal{X}}

\newcommand{\abs}[1]{\left\vert#1\right\vert}
\newcommand{\set}[1]{\left\{#1\right\}}

\newcommand{\calI}{\mathcal{I}}

\newcommand{\eps}{\varepsilon}

\newcommand{\R}{\mathds{R}}
\newcommand{\Na}{\mathds{N}}

\newcommand{\Exp}{\mathds{E}}

\newcommand{\Prob}{\mathbf{P}}

\newcommand{\B}{\mathcal{B}}
\newcommand{\A}{\mathcal{A}}

\newcommand{\Fil}{\mathds{F}}

\newcommand{\Qbf}{\mathbf{Q}}

\usepackage[auth-lg]{authblk}

\title{Martingale approach to optimal portfolio-consumption problems in Markov-modulated pure-jump models}

\author{\large Oscar L\'opez\thanks{oscar.lopez@urosario.edu.co} }

\author{\large Rafael Serrano\thanks{rafael.serrano@urosario.edu.co} \thanks{The authors gratefully acknowledge the financial support of FIUR research project  DVG170}}

\affil{\small Universidad del Rosario\\
Calle 12C No. 4-69\\
Bogot\'a, Colombia}

\begin{document}

\bibliographystyle{amsplain}
\maketitle

\begin{abstract}
  We study optimal investment strategies that maximize expected utility from consumption and terminal wealth in a pure-jump asset price model with Markov-modulated (regime switching) jump-size distributions. We give sufficient conditions for existence of optimal policies and find closed-form expressions for the optimal value function for agents with logarithmic and fractional power (CRRA) utility in the case of two-state Markov chains. The main tools are convex duality techniques, stochastic calculus for pure-jump processes and explicit formulae for the moments of telegraph processes with Markov-modulated random jumps.

\end{abstract}
\section{Introduction}
The object of this paper is to study the problem of maximizing expected utility from consumption and terminal wealth in an incomplete pure-jump asset price model with jump-size distributions modulated by an underlying continuous-time finite-state Markov chain, and totally inaccessible jump times that coincide with the transition times of the Markov chain. This financial market model is an extension of the jump-telegraph model proposed by L\'opez and Ratanov \cite{lopezrat} to the case in which the underlying Markov chain has more than two states. This generalization is mainly motivated by the empirical results of Konikov and Madan \cite{konmadan} that suggest that more than two regimes should be considered.

Markov-modulated regime-switching models have attracted considerable attention in financial modelling in the past 15 years as they allow for time-inhomogeneity in the asset dynamics that capture important features of financial time series such as asymmetric and heavy-tailed asset returns, time-varying conditional volatility and volatility clustering, as well as structural changes in economics conditions.

Our approach to the portfolio-consumption problem is largely based on the martingale approach and convex duality techniques for utility maximization in incomplete markets initiated by He and Pearson \cite{hepearson}, Karatzas et al. \cite{karatzas91}, Cvitani\'{c} and Karatzas \cite{cvikar}, and extended by Kramkov and Schachermayer \cite{kramsch} to the general semi-martingale setting. In the particular case of market models driven by jump-diffusion, Goll and Kallsen \cite{goll2000}, Kallsen \cite{kallsen2000} and more recently Michelbrink and Le \cite{michelbrink}, use the martingale approach to obtain explicit solutions for agents with logarithmic and power utility functions. Callegaro and Vargiolu \cite{callevar} obtain similar results in jump-diffusion models with Poisson-type jumps.

To the best of our knowledge, this is the first paper that uses the martingale approach to study the problem of maximizing utility from consumption and terminal wealth in a pure-jump model with Markov-modulated jumps. Only the optimal investment problem for jump-diffusion models studied  by B\"{a}uerle and Riedler \cite{bauerle2007} seems comparable to the formulation of the problem in the present paper. They use, however, the standard dynamic programming approach and are only able to derive some bounds on the optimal policy. Moreover, they assume that the jump-sizes do not depend on the underlying Markov chain.

The main result of this paper is a sufficient condition for existence of an optimal portfolio-consumption pair. This condition is given in terms of the solution pair of a linear backward SDE with respect to the compensated (random) counting measure associated with the marked point process consisting of the jump times and the corresponding (Markov-modulated) jumps. The coefficients of the backward SDE are related to the state-price densities via the convex conjugate of the utility functions. Although the optimality condition in the main result seems rather restrictive, in the last section we show that it simplifies significantly in the case of logarithmic utility functions.

The key assumption throughout is that the compensator of the counting measure has an (intensity) kernel that is also Markov-modulated, similar to the model proposed recently by Elliott and Siu \cite{elliottsiu}. In fact, using results on jump-telegraph processes with Markov-modulated jumps, we prove that if the underlying Markov chain takes only two values, our market model actually satisfies the main assumption, hence generalizing the model of Elliott and Siu \cite{elliottsiu}.

The outline of the paper is as follows. In Section 2 we describe the stochastic setting and information structure for the market model, introduce the wealth equation and define the optimal investment problem. In Section 3, following arguments similar to Michelbrink and Le \cite{michelbrink}, we formulate and prove the main result of this paper.  In section 4 we present some basic properties of the telegraph process with Markov-modulated jumps and prove that our market model satisfies the main assumption in the case of two regimes. Finally, Section 5 illustrates the main result by considering the special case of agents with logarithmic and fractional power (CRRA) utility. We also present some numerical results for the case of logarithmic utility and, using results from Section 4, we find a closed-form solution for the optimal value function in the case of two regimes.

\section{Market model, wealth equation and portfolio-consumption problem}
In this section we describe the stochastic setting and information structure for the market model  and introduce the wealth process and utility maximization problem from consumption and terminal wealth.

We fix an finite investment horizon $T>0.$ Let $\eps(\cdot)=\{\eps(t)\}_{t\in[0,T]}$ be a continuous-time Markov chain with finite state-space $\mathcal{I}=\{0,1,\ldots,m-1\}.$ Let $\set{\tau_n}_{n\geq 1}$ denote the jump times of the Markov chain $\eps(\cdot),$ and let $\eps_n:=\eps(\tau_n-)$ denote the state of $\eps(\cdot)$ right before the $n$-th jump.

Let $E$ be an Euclidean space with Borel $\sigma$-algebra $\B(E).$ For each $i\in\mathcal{I}$, let $\{Y_{i,n}\}_{n\geq 1}$ denote a sequence of $E$-valued independent random variables with distributions
\[
\Prob(Y_{i,n}\in dy)={F}_i(dy), \ \ n\geq 1.
\]
For each $i\in\mathcal{I},$ the event $\eps(t)=i$ represents that the economy or business cycle is in the $i$-th state at time $t.$ Let $(r_i)_{i\in\calI}$ with $r_i>0$ and $(\mu_i)_{i\in\calI}$ denote the vectors of instantaneous interest rates and stock appreciation rates in each state or regime. The financial market consists of a default-free money-market account with Markov-modulated continuously compounded return rate $\{r_{\eps(t)}\}_{t\in[0,T]},$ that is, its price process $B=\{B_t\}_{t\in[0,T]}$ satisfies
\[
B_t=\exp\left(\int_0^tr_{\eps(s)}ds\right), \ \ t\in[0,T],
\]
and a risky asset or stock with price process $S=\{S_t\}_{t\in[0,T]}$ solution of the equation
\begin{equation}\label{eqS0}
dS_t=S_{t-}\,dX_t
\end{equation}
where $X=\{X_t\}_{t\in[0,T]}$ is the pure-jump process with Markov-modulated jumps
\[
X_t=\int_0^t\mu_{\eps(s)}\,ds+\sum_{\tau_n\le t}f_{\eps_n}(Y_{\eps_n,n}), \ \ t\in[0,T].
\]
For each $i\in\calI,$ $f_i:E\rightarrow (-1,\infty)\setminus\set{0}$ is a measurable map, integrable with respect to the distribution $F_i(dy).$ Throughout, we assume that the distributions $F_i$ are pairwise independent, as well independent of the Markov chain $\eps(\cdot).$

Let $\gamma:\Omega\times \B(E)\otimes\B([0,T])\to\set{1,2,\ldots}$ be the random counting measure associated with the  sequence $\set{(\tau_n,Y_{\eps_n,n})}_{n\geq 1}$ defined as
\[
\gamma(A\times(0,t]):=\sum_{n=1}^\infty\mathbf{1}_{\set{\tau_n\le t,Y_{\eps_n,n}\in A}}=\sum_{\tau_n\le t}\mathbf{1}_{\set{Y_{\eps_n,n}\in A}}, \ \ A\in\mathcal{B}(E), \ \ t\in [0,T]
\]
The random measure $\gamma$ is known as marked point process or multivariate point process with mark space $E,$ see e.g. Jacod and Shiryaev \cite[Chapter III, Definition 1.23]{js} or Jeanblanc et al \cite[Section 8.8]{jeanblanc}.

For each $A\in\B(E),$ the counting process $N_t(A):=\gamma(A\times(0,t])$ counts the number of marks with values in $A$ up to time $t.$ The underlying filtration $\Fil=\set{\calF_t}_{t\in  [0,T]}$ is defined as the filtration generated by these counting processes, augmented with the $\sigma$-algebra $\calF_0$ of $\Prob$-null events,
\[
\calF_t:=\calF_0\vee\sigma(N_s(A):0\le s\le t, \ A\in\B(E)), \ \ t\in  [0,T].
\]
The predictable $\sigma$-algebra $\calP$ on $\Omega\times  [0,T]$ is defined as the $\sigma$-algebra generated by adapted left-continuous processes. A real-valued process $\set{\phi_t}_{t\in  [0,T]}$ is said to be $\mathds{F}$-predictable if the random function $\phi(t,\omega)=\phi_t(\omega)$ is measurable with respect to $\calP$.

Similarly, a map $\phi:\Omega\times  [0,T]\times E\to\R$ is said to be $\Fil$-predictable if it is measurable with respect to the product $\sigma$-algebra $\calP\otimes\B(E).$ For $\phi$ $\Fil$-predictable, we may define the stochastic integral of $\phi$ with respect to the random measure $\gamma(dy,dt)$ as follows
\begin{equation*}
\int_{(0,t]}\int_{E} \phi(s,y)\gamma(dy,ds)
:=\sum_{\tau_n\leq t}\phi(\tau_n,Y_{\eps_n,n})=\sum_{n=1}^{N_t(E)}\phi(\tau_n,Y_{\eps_n,n}), \ \ t\in [0,T].
\end{equation*}
Using this definition, we can rewrite equation (\ref{eqS0})  as
\begin{equation}\label{eqS}
dS_t={S_{t-}}\left[\mu_{\eps(t)}\,dt+\int_{E} f_{\eps(t-)}(y)\,\gamma(dy,dt)\right].
\end{equation}
The solution of this linear equation is given by the process
\begin{align*}
S_t&=S_0\mathcal{E}_t\left(\int_0^\cdot\mu_{\eps(s)}\,ds+\int_0^\cdot\int_E f_{\eps(s-)}(y)\,\gamma(dy,ds) \right)\\
&=S_0\exp\left(\int_0^t\mu_{\eps(s)}\,ds\right)\prod_{n=1}^{N_t(E)}(1+f_{\eps_n}(Y_{\eps_n,n})), \ \ t\in [0,T].
  \end{align*}
Here, $\mathcal{E}_t(\cdot)$ denotes the stochastic (Dol\'{e}ans-Dade) exponential, see e.g. Jeanblanc et al \cite[Section 9.4.3]{jeanblanc}. Moreover, since $f_i(y)>-1,$ the price process $S$ satisfies
\begin{align*}
S_t&=S_0\exp\biggl(\int_0^t\mu_{\eps(s)}\,ds+\sum_{n=1}^{N_t(E)}\ln(1+f_{\eps_n}(Y_{\eps_n,n}))\biggr)\\
&=S_0\exp\left(\int_0^t\mu_{\eps(s)}\,ds+\int_0^t\int_{E}\ln(1+f_{\eps(s-)}(y))\,\gamma(dy,ds)\right), \ \ t\in [0,T].
\end{align*}
 The log-returns of the stock process $S_t$ are then given by the pure-jump process with Markov-modulated random jumps
\[
\ln\Bigl(\frac{S_t}{S_0}\Bigr)=\int_0^t\mu_{\eps(s)}\,ds+\sum_{n=1}^{N_t(E)}\ln(1+f_{\eps_n}(Y_{\eps_n,n})), \ \ t\in [0,T].
\]
For an agent willing to invest in the financial market described above, let $\pi_t$ denote the fraction of wealth invested in the risky asset $S_t$ at time $t,$ so that the fraction of wealth invested in the money account $B_t$ is $1-\pi_t.$ Recall that a positive value for $\pi_t$ represents a long position in the risky asset, whereas a negative $\pi_t$ stands for a short position.

During the time interval $[0,T],$ the investor is allowed to consume at an instantaneous consumption rate $c_t.$ In the following, we consider only portfolio-consumption pairs $(\pi,c)=\set{(\pi_t,c_t)}_{t\in [0,T]}$ that are $\Fil$-predictable, satisfy the integrability condition
\[
\int_0^T(\pi_t^2+c_t)\,dt<+\infty, \ \ \mbox{a.s.},
\]
as well as the so-called \emph{self-financing condition}, that is, for an initial wealth $x>0$ and a portfolio-consumption pair $(\pi,c),$ the wealth $V_t$ at time $t$ of the investor satisfies the stochastic differential equation
\begin{equation}\label{eqV}
\begin{split}
dV_t&=(V_{t-}r_{\eps(t)}-c_t)\,dt+\pi_tV_{t-}\left\{(\mu_{\eps(t)}-r_{\eps(t)})\,dt+\int_{E} f_{\eps(t-)}(y)\,\gamma(dy,dt)\right\},\\
V_0&=x.
\end{split}
\end{equation}
We denote with $V^{x,\pi,c}=\{V_t^{x,\pi,c}\}_{t\in [0,T]}$  the solution to equation (\ref{eqV}). In particular, if there is no consumption i.e. $c_t =0$ for all $t\in[0,T],$ equation (\ref{eqV}) is linear and its solution is given explicitly by
\begin{align}
V_t^{x,\pi,0}&=x\mathcal{E}_t\left(\int_0^\cdot \left[r_{\eps(s)}+\pi_s\left(\mu_{\eps(s)}-r_{\eps(s)}\right)\right]\,ds+\int_0^\cdot\int_E\pi_s f_{\eps(s-)}(y)\,\gamma(dy,ds)\right)\notag\\
&=x\exp\left(\int_0^t\left[\pi_s\mu_{\eps(s)}
+(1-\pi_s)r_{\eps(s)}\right]\,ds\right)\prod_{n=1}^{N_t(E)}(1+\pi_{\tau_n}f_{\eps_n}(Y_{\eps_n,n})).\label{Vpi0}
\end{align}
Notice that this is always positive if, for instance, short-selling is not allowed for any of the assets i.e. if $\pi_t\in [0,1]$ for all $t\in[0,T].$
We can use (\ref{Vpi0}) to find an expression for the wealth process $V^{x,\pi,c}$ in terms of $V^{1,\pi,0}=\{V_t^{1,\pi,0}\}_{t\in [0,T]},$ the wealth process with initial wealth $1$ and portfolio-consumption pair $(\pi,0),$ as follows:  consider the process
\[
\xi_t^{x,\pi,c}:=x-\int_0^t\frac{c_s}{V_{s-}^{1,\pi,0}}\,ds, \ \ t\in [0,T].
\]
In differential form, we have $V_{t-}^{1,\pi,0}\,d\xi_t^{{x,\pi,c}}=-c_t\,dt.$ Then
\begin{align*}
d\left(\xi_t^{{x,\pi,c}}V_t^{1,\pi,0}\right)&=\xi_t^{{x,\pi,c}}\,dV_t^{1,\pi,0}+V_{t-}^{1\pi,0}\,d\xi_t^{{x,\pi,c}}\\
&=\xi_t^{{x,\pi,c}}V_{t-}^{1,\pi,0}\left(\left[r_{\eps(t)}+\pi_t\left(\mu_{\eps(t)}-r_{\eps(t)}\right)\right]\,dt+\pi_t\int_{E} f_{\eps(t-)}(y)\,\gamma(dy,dt)\right)-c_t\,dt.
\end{align*}
Since $\xi_0^{{x,\pi,c}}V_0^{1,\pi,0}=x,$ by uniqueness of solution to equation (\ref{eqV}), the wealth process $V^{x,\pi,c}$ is a modification of
\begin{equation}\label{xiV}
  \begin{split}
&\xi_t^{{x,\pi,c}}V_t^{1,\pi,0}=\\
&\phantom{\gamma t}\biggl[x-\int_0^t\frac{c_s}{V_{s-}^{1,\pi,0}}\,ds\biggr]\exp\left(\int_0^t\left[\pi_s\mu_{\eps(s)}
+(1-\pi_s)r_{\eps(s)}\right]\,ds\right)\prod_{n=1}^{N_t(E)}(1+\pi_{\tau_n}f_{\eps_n}(Y_{\eps_n,n})).
\end{split}
\end{equation}
Notice that the portfolio-consumption pair $(\pi,c)$ leads to positive wealth at time $t\in [0,T]$ if, almost surely
\[
\int_0^t\frac{c_s}{V_{s-}^{1,\pi,0}}\,ds<x \ \ \mbox{ and } \ \ \pi_{\tau_n}f_{\eps_n}(Y_{\eps_n,n})>-1, \ \forall \tau_n\le t.
\]
The class $\calA(x)$ of admissible pairs for initial wealth $x>0$ is defined as the set of portfolio-consumption pairs $(\pi,c)$  for which equation (\ref{eqV}) possesses an unique strong solution $V^{x,\pi,c}=\{V_t^{x,\pi,c}\}_{t\in [0,T]}$ such that $V_t^{x,\pi,c}\geq 0,$ a.s. for all $t\in[0,T].$

We now define the utility maximization problem for optimal choice of portfolio and consumption processes.
Let $U_1 : [0,T] \times [0,\infty) \rightarrow [-\infty,\infty)$ and $U_2 : [0,\infty) \rightarrow [-\infty,\infty)$ denote consumption and
investment utility functions respectively, satisfying the following conditions
%
\begin{enumerate}
\item[(i)] $U_1(t,x) > -\infty$
and $U_2(x) > -\infty$
for all $t\in [0,T]$
and $x \in (0,\infty)$
%
%

\item[(ii)] for each $t \in [0,T]$ the mappings
$U_1(t,\cdot): (0,\infty) \rightarrow \R$
and $U_2(\cdot) : (0,\infty) \rightarrow \R$
are strictly increasing, strictly concave, of class $C^{1}$ on $(0,\infty),$
such that
\[
\lim_{x \downarrow 0, \; x > 0} \frac{\partial U_1}{\partial x}(t,x) = +\infty,\;\;\;
\lim_{x \rightarrow \infty}  \frac{\partial U_1}{\partial x}(t,x) = 0,\;\;\;
\lim_{x \downarrow 0, \; x > 0} U_2'(x) = +\infty,\;\;\;
\lim_{x \rightarrow \infty}  U_2'(x) = 0.
\]
\item[(iii)] $U_1$ and
$\frac{\partial U_1}{\partial x}$
are continuous on $[0,T]\times (0,\infty)$.
\end{enumerate}%

%
Given the initial state of the Markov chain $\eps(0)=i\in\calI,$ let $\tilde{\calA}_i(x)$ denote the class of admissible portfolio-consumption strategies $(\pi,c)\in\calA(x)$ such that
\[
\Exp_i\left[\int_0^T U_1(t,c_t)^{-}\,dt+U_2(V_T^{x,\pi,c})^{-}\right]>-\infty,
\]
where $\alpha^{-}:=\min\set{0,\alpha}$ is the negative part of $\alpha\in\R$ and $\Exp_i[\cdot]:=\Exp[\cdot|\eps(0)=i].$ We define the utility functional
\[
J_i(x;\pi,c):=\Exp_i\left[\int_0^T U_1(t,c_t)\,dt+U_2(V_T^{x,\pi,c})\right], \ \ (\pi,c)\in\tilde{\calA}_i(x),
\]
and consider the following utility maximization problem from terminal wealth and consumption
\begin{equation}\label{utilitymax}
\vartheta_i(x):=\sup_{(\pi,c)\in\tilde{\calA}_i(x)}J_i(x;\pi,c), \ \ x>0, \ \ i\in\calI.
\end{equation}
An admissible portfolio-consumption pair $(\hat{\pi},\hat{c})\in\tilde{\calA}_i(x)$ is said to be optimal for the initial state $\eps(0)=i$ and initial wealth $x>0$ if $\vartheta_i(x)=J_i(x;\hat{\pi},\hat{c}).$

\section{Martingale approach and main result}
Recall that the compensator $\rho(dy,dt)$ of the marked point process $\gamma(dy,dt)$ is the unique (possibly, up to a null set) predictable random measure such that, for every $\Fil-$predictable map $\phi(t,y)$ the two following conditions hold
\begin{itemize}
  \item[i.] The process
  \[
  \int_0^t\int_E\phi(s,y)\,\rho(dy,ds), \ \ t\geq 0,
  \]
  is $\Fil-$predictable.

  \item[ii.] If the process
  \[
  \int_0^t\int_E\abs{\phi(s,y)}\,\rho(dy,ds)<+\infty, \ \ \forall t\geq 0.
\]
is increasing and locally integrable, then
\[
M_t(\phi):=\int_0^t\int_E \phi(s,y)\,\gamma(dy,ds)-\int_0^t\int_E\phi(s,y)\,\rho(dy,ds), \ \ t\geq 0,
\]
is $\Fil$-local martingale (see e.g. Jeanblanc et al \cite[Definition 8.8.2.1]{jeanblanc}).
\end{itemize}
The following is the main assumption for the rest of this section
\begin{Assumption}\label{A1}
There exists $(\lambda_i)_{i\in\calI}$ with $\lambda_i>0$ for all $i\in\calI$ such that the compensator $\rho(dy,dt)$ of $\gamma(dy,dt)$ satisfies
  \begin{equation}\label{condrho}
  \rho(dy,dt)=F_{\eps(t-)}(dy)\lambda_{\eps(t-)}\,dt, \ \mbox{ a.s.}
  \end{equation}
\end{Assumption}

\begin{remark}
Condition (\ref{condrho}) is similar to the main assumption in the recent paper by Elliott and Siu \cite{elliottsiu}. In the next section, using properties of jump-telegraph processes, we will prove that Assumption \ref{A1} actually holds for our market model in the case of a two-state Markov chain.
\end{remark}

\begin{remark}
Under Assumption \ref{A1}, the counting process $N_t(A):=\gamma(A\times (0,t])$ is an inhomogeneous Poisson process with stochastic (Markov-modulated) intensity $F_{\eps(t-)}(A)\lambda_{\eps(t-)},$ see e.g. Jeanblanc et al \cite[Section 8.4.2]{jeanblanc} or the proof of Corollary \ref{lemrhodis} below.
\end{remark}
Let $\widetilde{\gamma}(dy,dt):=\gamma(dy,dt)- F_{\eps(t-)}(dy)\lambda_{\eps(t-)}\,dt$ denote the compensated martingale measure associated with the counting measure $\gamma(dy,dt).$ We define equivalent martingale probability measures via the Radon-Nikodym densities
\[
\left.\frac{d\Qbf^\varphi}{d\Prob}\right|_{\calF_t}:=Z_t^\varphi, \ \ t\in[0,T],
\]
where $Z^\varphi=\{Z_t^\varphi\}_{t\in[0,T]}$ is the solution of the linear SDE
\begin{equation}\label{eq:dZt}
dZ_t=Z_{t-}\int_{E}(\varphi_{\eps(t-)}(y)-1)\,\widetilde{\gamma}(dy,dt), \ \ Z_0=1,
\end{equation}
and, for each $i\in\calI,$ $\varphi_i:\Omega\times E\to (0,\infty)$ is a nonnegative-valued $\Fil$-predictable map satisfying $\int_E\varphi_{i}(y)\,F_{i}(dy)<+\infty$ a.s. In what follows, we denote $\varphi=(\varphi_i)_{i\in\calI}.$

If the process $Z^\varphi$  satisfies $\Exp[Z_T^\varphi]=1,$ then $Z^\varphi$ is a $\Fil$-martingale under $\Prob,$ and $\Qbf^\varphi$ defines a probability measure on $(\Omega,\calF_T),$ see e.g. Theorem T10 in Br\'{e}maud \cite[Chapter VIII]{bremaud}. Moreover, the compensator measure $\rho^\varphi(dy,dt)$ of $\gamma(dy,dt)$ under $\Qbf^\varphi$ satisfies

\[
\rho^\varphi(dy,dt)=F_{\eps(t-)}^\varphi(dy)\lambda_{\eps(t-)}^\varphi\,dt
\]
where, for each $i\in\calI,$
\begin{align*}
F_{i}^\varphi(dy)&:=\frac{\varphi_{i}(y)}{h_{i}^\varphi}F_{i}(dy),\\
\lambda_{i}^{\varphi}&:=\lambda_{i}h_{i}^\varphi,
\end{align*}
and
\[
h_{i}^\varphi:=\int_{E}\varphi_{i}(y)\,F_{i}(dy).
\]
\begin{remark}\label{martcondZ}
If
\begin{equation}\label{alpha}
\int_E\varphi_i(y)^p\,F_i(dy)<+\infty, \ \ \mbox{a.s.}, \ \forall i\in\mathcal{I}
\end{equation}
for some $p>1,$ then $\Exp[Z_T^\varphi]=1$ and $\Qbf^\varphi$ defines a probability measure on $(\Omega,\calF_T),$ see e.g. Theorem T11 of Br\'{e}maud \cite[Chapter VIII]{bremaud}. For $m=2,$ the same holds if condition (\ref{alpha}) is satisfied with $p=1,$ see Remark \ref{remZ:m2} below.
\end{remark}
Let $\Theta$ denote the set of $m$-tuples $\varphi=(\varphi_i)_{i\in\calI}$ of non-negative valued $\Fil$-predictable maps for which $Z^\varphi$ is a $\Fil$-martingale and the following condition holds
  \begin{equation}\label{martcond}
  \mu_i-r_i+\lambda_i\int_{E} f_i(y)\varphi_i(y)\,F_i(dy)=0, \ \ \mbox{a.s. for all }  i\in\calI.
  \end{equation}
\begin{remark}
Under condition (\ref{martcond}), the discounted asset price $\{B_t^{-1}S_t\}_{t\in[0,T]}$ is a $\Qbf^\varphi$-local martingale. Indeed, let $\widetilde{\gamma}^{\,\varphi}(dy,dt):=\gamma(dy,dt)-F_{\eps(t-)}^\varphi(dy)\lambda_{\eps(t-)}^\varphi\,dt$ denote the compensated martingale measure associated under the equivalent probability measure $\Qbf^\varphi.$ Then, we have
\begin{align*}
  d&(B_t^{-1}S_t)=B_t^{-1}S_{t-}\Bigl\{\left(\mu_{\eps(t)}-r_{\eps(t)}\right)\,dt+\int_{E} f_{\eps(t-)}(y)\,\gamma(dy,dt)\Bigr\}\\
  &=B_t^{-1}S_{t-}\left\{\Bigl[\mu_{\eps(t-)}-r_{\eps(t-)}+\lambda_{\eps(t-)}^\varphi\int_{E} f_{\eps(t-)}(y)\,F_{\eps(t-)}^\varphi(dy)\Bigr]\,dt+\int_{E} f_{\eps(t-)}(y)\,\widetilde{\gamma}^{\,\varphi}(dy,dt)\right\}\\
  &=B_t^{-1}S_{t-}\int_{E} f_{\eps(t-)}(y)\,\widetilde{\gamma}^{\,\varphi}(dy,dt)
\end{align*}
and the claim follows. Moreover, since $f_i(y)\neq 0$ for $i\in\calI$ the market model is arbitrage-free (see e.g. Remark 1 of Ratanov and Melnikov \cite{ratmel} for the case of two-regimes and deterministic jumps). Then, by the first fundamental theorem of asset pricing, the set $\Theta$ is non-empty.
\end{remark}

For each $\varphi\in\Theta$ we define the \emph{state price density} process $H^\varphi=\{H_t^{\varphi}\}_{t\in[0,T]}$ by
\[
H_t^{\varphi}:=B_{t}^{-1}Z_t^\varphi, \ \ t\in [0,T].
\]
The solution of equation (\ref{eq:dZt}) is given by
\begin{align*}
Z_t^\varphi&=\exp\left(\int_0^t\int_{E}\bigl(1-\varphi_{\eps(s)}(y)\bigr)\lambda_{\eps(s)}F_{\eps(s)}(dy)ds+\int_0^t\int_{E}\ln \varphi_{\eps(s-)}(y)\gamma(dy,ds)\right)\\
&=\exp\left(\int_0^t\bigl(1-h_{\eps(s)}^{\varphi}\bigr)\lambda_{\eps(s)}\,ds+\int_0^t\int_{E}\ln \varphi_{\eps(s-)}(y)\gamma(dy,ds)\right).
\end{align*}
Then, the process $H^{\varphi}$ satisfies
\begin{equation*}
H_t^{\varphi}
=\exp\left(\int_0^t\bigl[\bigl(1-h_{\eps(s)}^{\varphi}\bigr)\lambda_{\eps(s)}-r_{\eps(s)}\bigr]ds+\int_0^t\int_{E}\ln\bigl(\varphi_{\eps(s-)}(y)\bigr)\gamma(dy,ds)\right), \ \ t\in [0,T].
\end{equation*}
The following is a well-known result for the state price density process $H^\varphi$ usually referred to as \emph{budget constraint.} We include the proof for the sake of completeness.

\begin{proposition}
For all $\varphi\in\Theta$ and $(\pi,c)\in\A(x),$ we have
\begin{equation}\label{budget}
\Exp_i\left[H_T^\varphi V_T^{x,\pi,c}+\int_0^TH_s^\varphi c_s\,ds\right]\le x, \ \ \forall i\in\calI.
\end{equation}
\end{proposition}
\begin{proof}
Using the product rule for jump processes     and (\ref{martcond}), we have
\begin{align*}
  d&(H^\varphi_t V^{x,\pi,c}_t)+H^\varphi_tc_t\,dt\\
  &=H^\varphi_{t-}dV^{x,\pi,c}_t+V^{x,\pi,c}_{t-}dH^\varphi_{t}+H^\varphi_{t-}V^{x,\pi,c}_{t-}\pi_t\int_Ef_{\eps(t-)}(y)(\varphi_{\eps(t-)}(y)-1)\,\gamma(dy,dt)+H^\varphi_tc_t\,dt\\
  &=H^\varphi_{t-}V^{x,\pi,c}_{t-}\left\{r_{\eps(t)}\,dt+\pi_t\Bigl[(\mu_{\eps(t)}-r_{\eps(t)})\,dt+\int_E f_{\eps(t-)}(y)\,\gamma(dy,dt)\Bigr]\right\}-H^\varphi_tc_t\,dt\\
  &\phantom{=}+H^\varphi_{t-}V^{x,\pi,c}_{t-}\left\{-r_{\eps(t)}\,dt+\int_E(\varphi_{\eps(t-)}(y)-1)\,\tilde{\gamma}(dy,dt)\right\}\\
  &\phantom{=}+H^\varphi_{t-}V^{x,\pi,c}_{t-}\pi_t\int_Ef_{\eps(t-)}(y)(\varphi_{\eps(t-)}(y)-1)\,\gamma(dy,dt)+H^\varphi_tc_t\,dt\\
  &=H^\varphi_{t-}V^{x,\pi,c}_{t-}\int_E(\pi_tf_{\eps(t-)}(y)\varphi_{\eps(t-)}(y)+\varphi_{\eps(t-)}(y)-1)\,\tilde{\gamma}(dy,dt)
\end{align*}
Integrating, we get
\begin{equation*}
  H_tV_t+\int_0^t H_s^\varphi c_s\,ds
  = x+ \int_0^t\int_E(\pi_sf_{\eps(s-)}(y)\varphi_{\eps(s-)}(y)+\varphi_{\eps(s-)}(y)-1)\,\tilde{\gamma}(dy,ds)
\end{equation*}
almost surely, for all $t\in [0,T].$ The stochastic integral in the right hand side is a $\Fil$-local martingale which is bounded below, hence a $\Fil-$super martingale, and (\ref{budget}) follows.
\end{proof}
We now introduce an auxiliary functional related to the convex dual of the utility functions. Let $U$ denote either $U_2(\cdot)$ or $U_1(t,\cdot)$ with $t\in[0,T]$ fixed. Let $I$  denote the inverse of $U',$  so that
\[
I(U'(x))=U'(I(x))=x, \ \ \forall x>0.
\]
Then, $I$ satisfies
\[
I(y)=\arg\max_{x>0}\set{U(x)-yx}, \ \ y>0.
\]
In particular,
\begin{equation}\label{ineqUI}
U(I(y))-yI(y)\geq U(x)-yx, \ \ \forall x,y>0.
\end{equation}
Notice that $U(I(y))-yI(y)=U^*(y),$ where $U^*(y):=\sup_{x>0}\set{U(x)-yx}$ is the Legendre-Fenchel transform of the map $(-\infty,0)\ni x\mapsto -U(-x)\in\R.$ The map $U^*$ is known as the convex dual of the utility function $U.$


For the rest of this section we fix the initial regime $\eps(0)=i\in\calI.$ For $\varphi\in\Theta,$ we define the map
\[
\calX_i^\varphi(y):=\Exp_i\left[\int_0^T H^\varphi_tI_1(t,yH_t^\varphi)\,dt+H^\varphi_TI_2(yH^\varphi_T)\right].
\]
Let $\widetilde{\Theta}=\set{\varphi\in\Theta:\calX_i^\varphi(y)<\infty, \ \forall y>0, \ \forall i\in\calI}.$ For each $\varphi\in\widetilde{\Theta}$ we denote $\Y_i^\varphi:=(\calX_i^\varphi)^{-1}$ and define the process $c^{x,\varphi}=(c_t^{x,\varphi})_{t\in [0,T]}$ and random variable $G^{x,\varphi}$ as follows
\begin{equation}\label{cY}
\begin{split}
  c_t^{x,\varphi} &:=I_1(t,\Y_i^\varphi(x)H_t^\varphi), \ \ t\in [0,T],\\
  G^{x,\varphi}&:=I_2(\Y_i^\varphi(x)H_T^\varphi).
\end{split}
\end{equation}
Finally, we define the auxiliary functional
\[
L_i(x;\varphi):=\Exp_i\left[\int_0^T U_1(t,c_t^{x,\varphi})\,dt+U_2(G^{x,\varphi})\right], \ \ x>0, \ \varphi\in\widetilde{\Theta}.
\]
\begin{lemma}\label{ineqJL}
  For all $x>0$ and $i\in\calI,$ the following holds
     \[
   J_i(x;\pi,c)\leq L_i(x;\varphi), \ \forall (\pi,c)\in\tilde{\calA}_i(x), \ \varphi\in\widetilde{\Theta}.
   \]
\end{lemma}

\begin{proof}
From (\ref{ineqUI}) and (\ref{cY}), we have
\[
U_1(t,c_t)\le U_1(t,c_t^{x,\varphi})+\Y^{\varphi}(x)H_t^{\varphi}(c_t-c_t^{x,\varphi})
\]
and
\[
U_2(V_T^{x,\pi,c})\le U_2(G^{x,\varphi})+\Y^{\varphi}(x)H_T^{\varphi}(V_T^{x,\pi,c}-G^{x,\varphi}).
\]
Then, by (\ref{budget}) and the definition of $\Y_i^{\varphi},$ we have
\begin{align*}
 J_i(x;\pi,c) &\le L_i(x;\varphi)
 +\Y_i^{\varphi}(x)\cdot\Exp_i\left[\int_0^TH_t^{\varphi}(c_t-c_t^{x,\varphi})\,dt+H_T^\varphi(V_T^{x,\pi,c}-G^{x,\varphi})\right]\\
 &\le L_i(x;\varphi)+\Y_i^{\varphi}(x)[x-\calX_i^{\varphi}(\Y_i^{\varphi}(x))]\\
 &= L_i(x;\varphi)
\end{align*}
and the desired result follows.
\end{proof}



By the previous Lemma, we have $\vartheta_i(x)\le \tilde{\vartheta}_i(x),$ where $\tilde{\vartheta}_i(\cdot)$ is the optimal value function of the minimization problem
\begin{equation}\label{dual}
\tilde{\vartheta}_i(x):=\inf_{\varphi\in\widetilde{\Theta}}L_i(x;\varphi).
\end{equation}
In Theorem \ref{main} below, we find sufficient conditions to ensure $\vartheta_i(x)=\tilde{\vartheta}_i(x)$ as well as the existence of an optimal portfolio-consumption process $(\hat{\pi},\hat{c}).$

For each $x>0$ and $\varphi\in \widetilde{\Theta},$ consider the processes defined as
\[
Y_t^{x,\varphi}:=\Exp\left[\Bigl.H_T^\varphi G^{x,\varphi}+\int_t^TH_s^\varphi c_s^{x,\varphi}\,ds\,\Bigr|\calF_t\right], \ \ t\in [0,T],
\]
and
\[
M_t^{x,\varphi}  := Y_t^{x,\varphi}+\int_0^tH_s^{\varphi}c_s^{x,\varphi}\,ds, \ \ t\in [0,T].
\]
Observe that  $Y_0^{x,\varphi} = \calX_i^\varphi(\Y_i^\varphi(x))= x$ and $Y_t^{x,\varphi} \geq 0 $ for all  $t \in \left[0,T\right].$ Moreover, $M_t^{x,\varphi}$ satisfies
\begin{equation}
M_t^{x,\varphi}=\Exp\left[ H_T^{\varphi} G^{x,\varphi}+\Bigl.\int_0^T H_s^{\varphi} c_s^{x,\varphi}\, ds\, \Bigr|  \calF_t \right], \ t \in [0,T],
\end{equation}
that is, the process $M^{x,\varphi}=\set{M_t^{x,\varphi}}_{t\in[0,T]}$ is an $\Fil$-martingale. Let $\beta^{x,\varphi}(t,y)$ denote the essentially unique martingale representation coefficient of $ M_t^{x,\varphi}$ with respect to the compensated measure $\tilde{\gamma}(dy,dt),$   
\begin{equation}
M_t^{x,\varphi} = x+\int_0^t{\int_{E}} {\beta^{x,\varphi}(s,y)}\,{\tilde{\gamma}}(dy,ds), \ \ t\in [0,T],
\end{equation}
see e.g. Theorem T8 in Section VIII of Br\'{e}maud \cite{bremaud}. Then, the pair $(Y^{x,\varphi},\beta^{x,\varphi})$ satisfies the linear backward SDE
\begin{equation}\label{bsdeY}
Y_t^{x,\varphi}=H_T^\varphi G^{x,\varphi}+\int_t^T H_s^\varphi c_s^{x,\varphi}\,ds-\int_t^T\!\int_E\beta^{x,\varphi}(s,y)\,{\tilde{\gamma}}(dy,dt), \ \ t\in [0,T],
\end{equation}
with final condition $Y_T^{x,\varphi}=H_T^\varphi G^{x,\varphi}.$ The following is the main result of this paper

\begin{theorem}\label{main}
For $x>0$ and $i\in\calI$ fixed, suppose there exist $\hat{\varphi} \in \tilde{\Theta}$ and a $\Fil$-predictable portfolio process $\hat{\pi}$ satisfying
\begin{equation}\label{pif}
\hat{\pi}_tf_{\eps(t-)}(y)+1=\frac{1}{\hat{\varphi}_{\eps(t-)}(y)}\left[\frac{\beta^{x,\hat{\varphi}}(t,y)}{Y_{t-}^{x,\hat{\varphi}}}+1\right], \ \ \rho\mbox{-a.e. }(t,y)\in [0,T]\times E.
\end{equation}
Assume also that the wealth equation (\ref{eqV}) has a solution for $ (\hat{\pi} , \hat{c}),$ where $\hat{c}=c^{x,\hat{\varphi}}$. Then the following assertions hold

\begin{itemize}
   \item[(a)] The pair $(\hat{\pi} , \hat{c})$ belongs to $\tilde{\calA}_i(x)$ and solves the optimal portfolio-consumption problem (\ref{utilitymax}),
   \item[(b)] the wealth process $V^{x,\hat{\pi},\hat{c}}$ is a modification of the process $X_t^{x,\hat{\varphi}}:=Y_t^{x,\hat{\varphi}}/H_t^{\hat{\varphi}}$,
   \item[(c)] the optimal value function for the utility maximization (\ref{utilitymax}) satisfies $\vartheta_i=\mathcal{K}_i^{\hat{\varphi}}\circ\mathcal{Y}_i^{\hat{\varphi}}$ where
\[
\mathcal{K}_i^{\varphi}(y):=\Exp_i\left[\int_0^TU_1(t,I_1(t,yH_t^{\varphi}))\,dt+U_2(I_2(yH_T^{\varphi}))\right], \ \ y>0.
\]
 \end{itemize}
\end{theorem}
\begin{proof}
We prove first part (b). Since $X_0^{x,\hat{\varphi}}=Y_0^{x,\hat{\varphi}}=x,$ it suffices to show that $X_t^{x,\hat{\varphi}}$ satisfies the wealth equation for the pair $(\hat{\pi},\hat{c}).$ Notice first that, by the definition of $Z_t^{\hat{\varphi}},$ the process $H_t^{\hat{\varphi}}$ satisfies the linear stochastic equation
\begin{align*}
  dH^{\hat{\varphi}}_t&=H^{\hat{\varphi}}_{t-}\Bigl\{-r_{\eps(t)}\,dt+\int_E(\hat{\varphi}_{\eps(t-)}(y)-1)\,\tilde{\gamma}(dy,dt)\Bigr\}\\
          &=H^{\hat{\varphi}}_{t-}\Bigl\{\Bigl[-r_{\eps(t)}-\lambda_{\eps(t-)}\int_E(\hat{\varphi}_{\eps(t-)}(y)-1)\,F_{\eps(t-)}(dy)\Bigr]\,dt
          +\int_E(\hat{\varphi}_{\eps(t-)}(y)-1)\,\gamma(dy,dt)\Bigr\}
\end{align*}
Using integration formula for marked point processes (see e.g. Jeanblanc et al \cite[Section 8.8]{jeanblanc}), the differential of $1/H^{\hat{\varphi}}_t$ is given by
\[
  d\Bigl(\frac{1}{H^{\hat{\varphi}}_t}\Bigr)
  =\frac{1}{H^{\hat{\varphi}}_{t-}}\Bigl\{\Bigl[r_{\eps(t)}+\lambda_{\eps(t-)}\int_E(\hat{\varphi}_{\eps(t-)}(y)-1)\,F_{\eps(t-)}(dy)\Bigr]\,dt
  +\int_E\Bigl(\frac{1}{\hat{\varphi}_{\eps(t-)}(y)}-1\Bigr)\,\gamma(dy,dt)\Bigr\}.
\]
From (\ref{bsdeY}), the differential of $Y^{x,\hat{\varphi}}_t$ is given by
\[
dY^{x,\hat{\varphi}}_t=-H^{\hat{\varphi}}_tc^{x,\hat{\varphi}}_t\,dt+\int_E\beta^{x,\hat{\varphi}}(t,y)\,\tilde{\gamma}(dy,dt).
\]
Using the product rule, we have
\begin{align*}
  d\Bigl(\frac{Y^{x,\hat{\varphi}}_t}{H^{\hat{\varphi}}_t}\Bigr)
  &=Y_{t-}\,d\Bigl(\frac{1}{H^{\hat{\varphi}}_t}\Bigr)+\frac{1}{H^{\hat{\varphi}}_{t-}}\,dY^{x,\hat{\varphi}}_t
  +\frac{1}{H^{\hat{\varphi}}_{t-}}\int_E\beta^{x,\hat{\varphi}}(t,y)\Bigl(\frac{1}{\hat{\varphi}_{\eps(t-)}(y)}-1\Bigr)\,\gamma(dy,dt)\\
  &=\frac{Y_{t-}}{H^{\hat{\varphi}}_{t-}}\Bigl\{\Bigl[r_{\eps(t)}+\lambda_{\eps(t-)}\int_E(\varphi_{\eps(t)}(y)-1)\,F_{\eps(t)}(dy)\Bigr]\,dt
  +\int_E\Bigl(\frac{1}{\hat{\varphi}_{\eps(t-)}(y)}-1\Bigr)\,\gamma(dy,dt)\Bigr\}\\
  &\phantom{==}-c^{x,\hat{\varphi}}_t\,dt+\frac{1}{H^{\hat{\varphi}}_{t-}}\Bigl\{\int_E\beta^{x,\hat{\varphi}}(t,y)\,\tilde{\gamma}(dy,dt)
  +\int_E\beta^{x,\hat{\varphi}}(t,y)\Bigl(\frac{1}{\hat{\varphi}_{\eps(t-)}(y)}-1\Bigr)\,\gamma(dy,dt)\Bigr\}.
\end{align*}
We multiply and divide the last bracket by $Y_{t-}^{x,\hat{\varphi}}$ and use $\tilde{\gamma}(dy,dt)=\gamma(dy,dt)-\lambda_{\eps(t-)}\,F_{\eps(t-)}(dy)\,dt$ to obtain
\begin{align*}
d\Bigl(\frac{Y^{x,\hat{\varphi}}_t}{H^{\hat{\varphi}}_t}\Bigr)
  &=\Bigl\{\frac{Y^{x,\hat{\varphi}}_t}{H^{\hat{\varphi}}_t}r_{\eps(t)}-c^{x,\hat{\varphi}}_t\Bigr\}\,dt
+\frac{Y_{t-}}{H^{\hat{\varphi}}_{t-}}\left\{\lambda_{\eps(t-)}\int_E\Bigl(\varphi_{\eps(t-)}(y)-1-\frac{\beta^{x,\hat{\varphi}}(t,y)}{Y_{t-}}\Bigr)\,F_{\eps(t-)}(dy)\,dt\right.\\
&\phantom{==}\left.+\int_E\Bigl(\frac{1}{\hat{\varphi}_{\eps(t-)}(y)}-1+\frac{\beta^{x,\hat{\varphi}}(t,y)}{\hat{\varphi}_{\eps(t-)}(y)Y_{t-}}\Bigr)\,\gamma(dy,dt)\right\}
\end{align*}
From (\ref{pif}), for the integrand in the stochastic integral, we have
\[
\frac{1}{\hat{\varphi}_{\eps(t-)}(y)}-1+\frac{\beta^{x,\hat{\varphi}}(t,y)}{\hat{\varphi}_{\eps(t-)}(y)Y_{t-}}
  =\pi_tf_{\eps(t-)}(y)
\]
and (\ref{pif}) in conjunction with (\ref{martcond}) yields
\begin{align*}
\lambda_{\eps(t-)}&\int_E\Bigl(\hat{\varphi}_{\eps(t-)}(y)-1-\frac{\beta^{x,\hat{\varphi}}(t,y)}{Y_{t-}}\Bigr)\,F_{\eps(t-)}(dy)\\
&=\lambda_{\eps(t-)}\int_E(\hat{\varphi}_{\eps(t-)}(y)-\hat{\varphi}_{\eps(t-)}(y)(\pi_tf_{\eps(t-)}(y)+1))\,F_{\eps(t-)}(dy)\\
&=\lambda_{\eps(t-)}\int_E-\pi_t\hat{\varphi}_{\eps(t-)}(y)f_{\eps(t-)}(y)\,F_{\eps(t-)}(dy)\\
&=\pi_t\left(\mu_{\eps(t-)}-r_{\eps(t-)}\right)
\end{align*}
and part (b) follows. In particular, we have $V_T^{x,\hat{\pi},\hat{c}}=Y_T^{x,\hat{\varphi}}/H_T^{\hat{\varphi}}=G^{x,\hat{\varphi}},$ a.s. This in turn implies
\begin{equation}
L_i(x;\hat{\varphi})=J_i(x;\hat{\pi},\hat{c})
\end{equation}
and part (a) follows from Lemma \ref{ineqJL}. Part (c) follows easily since $\vartheta_i(x)=\tilde{\vartheta}_i(x)=L_i(x;\hat{\varphi}).$
\end{proof}
\begin{remark}
Although optimality condition (\ref{pif}) looks rather restrictive,  as we will see in the last section, it simplifies significantly in the case of logarithmic utility functions.
\end{remark}
\section{Telegraph processes with Markov-modulated random jumps}

In this section we revisit briefly the telegraph model with Markov-modulated random jumps introduced recently by L\'opez and Ratanov \cite{lopezrat}. We assume that the Markov chain $\eps(\cdot)=\{\eps(t)\}_{t\in[0,T]}$ takes only two values $\set{0,1}$ with intensity matrix $\bigl(\begin{smallmatrix}-\lambda_0 & \lambda_0\\ \lambda_1 &-\lambda_1\end{smallmatrix}\bigr).$ Thus, $X$ is given by the \emph{jump-telegraph process
}
\begin{equation}\label{Tt}
X_t=\int_0^t \mu_{\eps(s)}ds+\sum_{n=1}^{N_t(E)}f_{\eps_n}(Y_{\eps_n,n}), \ \ t\in [0,T].
\end{equation}
We assume that the alternating tendencies $\mu_0$ and $\mu_1$ satisfy $\mu_0\neq \mu_1.$ By fixing the initial state $\eps(0)=i\in\{0,1\}$, we have the following equality in distribution
\begin{equation}\label{eq:edZ}
X_t\overset{d}{=}\mu_i t\mathbf{1}_{\{t<\tau_1\}}+\bigl[\mu_i\tau_1+f_i(Y_{i,1})+\widetilde{X}_{t-\tau_1}\bigr]\mathbf{1}_{\{t>\tau_1\}}, \ \ t\in[0,T],
\end{equation}
where the process $\widetilde{X}=\{\widetilde{X}_t\}_{t\in[0,T]}$ is a jump-telegraph process as in (\ref{Tt}) independent of $X$  starting from the opposite initial state $1-i$.

\begin{figure}[htb]
\centering
\includegraphics{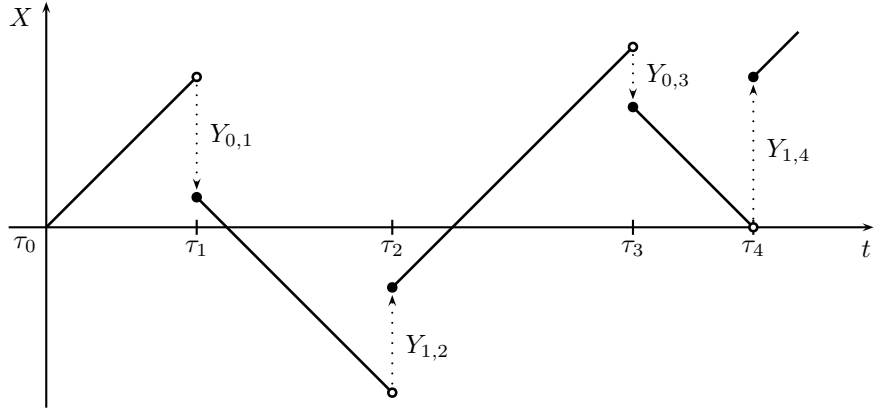}
\caption{\small A sample path of $X$ with $\mu_1<0<\mu_0$, $f_0(y)=f_1(y)=y$ and initial state $\eps(0)=0$.}
\end{figure}

We denote $\Prob_i(\cdot)=\Prob(\cdot|\eps(0)=i),$ and define $p_i(t,x)$ as the density function of the random variable $X_t,$ given the initial $\eps(0)=i\in\set{0,1},$
\begin{equation}
p_i(t,x):=\frac{\Prob_i(X_t\in d x)}{d x}, \ \ t\in[0,T], \ x\in\R.
\end{equation}
That is, for any $\Delta\in\B(\R)$, we have
\[
\int_{\Delta}\! p_{i}(t,x)\,dx=\Prob_i\left(X_t\in \Delta\right).
\]
Recall that the holding or inter-arrival times $\{\tau_{n+1}-\tau_{n}\}_{n\geq0}$ of the Markov chain $\eps(\cdot)$ are exponentially distributed with
\begin{equation}\label{eq:T-n}
\Prob(\tau_{n+1}-\tau_n>t\mid \mathcal{F}_{\tau_n}\bigr)=\exp\bigl(-\lambda_{\eps(\tau_n)}t\bigr).
\end{equation}
Here we have set $\tau_0:=0.$ Using (\ref{eq:T-n}) and (\ref{eq:edZ}) together with the total probability theorem, it follows that the densities functions $p_i(t,x)$ satisfy the following system of integral equations on $[0,T]\times\R,$
\begin{equation*}
\begin{split}
p_0(t,x)&=e^{-\lambda_0t}\delta(x-\mu_0t)+\int_0^t\left(\int_{E} p_{1}(t-s,x-\mu_0s-f_0(y))F_0(d y)\right)\lambda_0e^{-\lambda_0s}ds \\
p_1(t,x)&=e^{-\lambda_1t}\delta(x-\mu_1t)+\int_0^t\left(\int_{E} p_{0}(t-s,x-\mu_1s-f_1(y))F_1(d y)\right)\lambda_1e^{-\lambda_1s}ds
\end{split}
\end{equation*}
where $\delta(\cdot)$ is Dirac's delta function. This system is equivalent to the following system of coupled partial integro-differential equations on $(0,T]\times\R,$
\begin{equation}\label{eq:dX}
\begin{split}
\frac{\partial p_0}{\partial t}(t,x)+\mu_0\frac{\partial p_0}{\partial x}(t,x)&=-\lambda_0p_0(t,x)+\lambda_0\int_{E}p_{1}(t,x-f_0(y))F_0(dy)\\
\frac{\partial p_1}{\partial t}(t,x)+\mu_1\frac{\partial p_1}{\partial x}(t,x)&=-\lambda_1p_1(t,x)+\lambda_1\int_{E}p_{0}(t,x-f_1(y))F_1(dy)
\end{split}
\end{equation}
with initial conditions $p_0(0,x)=p_1(0,x)=\delta(x)$.
\begin{theorem}\label{thm:meanZG}
Assume $\int_{E} f_i(y)F_i(dy)<+\infty$ for $i=0,1.$ Then the conditional expectations $m_i(t):=\Exp_i[X_t]$ of the random variables $X_t$ satisfy
\begin{equation}\label{meanZG}
\begin{split}
m_0(t)&=\frac{1}{2\lambda}\left[(\lambda_1d_0+\lambda_0d_1)t+\lambda_0(d_0-d_1)\left(\frac{1-e^{-2\lambda t}}{2\lambda}\right)\right],\\
m_1(t)&=\frac{1}{2\lambda}\left[(\lambda_1d_0+\lambda_0d_1)t-\lambda_1(d_0-d_1)\left(\frac{1-e^{-2\lambda t}}{2\lambda}\right)\right],
\end{split}
\end{equation}
where
\begin{equation*}
2\lambda:=\lambda_0+\lambda_1,\quad \eta_i:=\int_{E} f_i(y)F_i(dy),\quad d_i:=\mu_i+\lambda_i\eta_i, \ \ i=0,1.
\end{equation*}
\end{theorem}
\begin{proof}
By definition, we have
\begin{equation*}
m_i(t)=\int_{-\infty}^\infty xp_i(t,x)d x,\quad i=0,1.
\end{equation*}
Differentiating the above equation, using the system \eqref{eq:dX} and integrating by parts, we obtain the following system of ODEs
\begin{equation*}
\begin{split}
\frac{d m_0}{d t}(t)&=-\lambda_0m_0(t)+\lambda_0m_1(t)+\mu_0+\lambda_0\eta_0,\\
\frac{d m_{1}}{d t}(t)&=-\lambda_1m_1(t)+\lambda_1m_0(t)+\mu_1+\lambda_1\eta_1,
\end{split}
\end{equation*}
with initial conditions $m_0(0)=m_1(0)=0$. The unique solution of this Cauchy problem is given by \eqref{meanZG}.
\end{proof}
\begin{theorem}\label{thm:expZG}
Assume $\int_{E} e^{f_i(y)}F_i(dy)<+\infty$ for $i=0,1.$ Then the conditional exponential moments $\psi_i(t):=\Exp_i[e^{X_t}]$ of the random variables $X_t$ satisfy
\begin{equation}\label{expZG}
\begin{split}
\psi_0(t)&=e^{t(\nu-\lambda)}\left[\cosh\bigl(t\sqrt{D}\,\bigr)
+\bigl(\mu-\zeta+\lambda_0\phi_0\bigr)\frac{\sinh\bigl(t\sqrt{D}\,\bigr)}{\sqrt{D}}\right],\\
\psi_1(t)&=e^{t(\nu-\lambda)}\left[\cosh\bigl(t\sqrt{D}\,\bigr)
-\bigl(\mu-\zeta-\lambda_1\phi_1\bigr)\frac{\sinh\bigl(t\sqrt{D}\,\bigr)}{\sqrt{D}}\right],
\end{split}
\end{equation}
where
\begin{equation*}
\mu:=\frac{\mu_0-\mu_1}{2},\quad \nu:=\frac{\mu_0+\mu_1}{2}, \quad \zeta:=\frac{\lambda_0-\lambda_1}{2},\quad
\phi_i:=\int_{E} e^{f_i(y)}F_i(dy), \ \ i=0,1
\end{equation*}
and $D=(\mu-\zeta)^2+\lambda_0\lambda_1\phi_0\phi_1$.
\end{theorem}
\begin{proof}
By definition, we have
\begin{equation*}
\psi_i(t)=\int_{-\infty}^\infty e^{x}p_i(t,x)d x,\quad i=0,1.
\end{equation*}
Differentiating the above equation, using the system \eqref{eq:dX} and integrating by parts, we obtain the following system of ODEs
\begin{equation*}
\begin{split}
\frac{d \psi_0}{d t}(t)&=(\mu_0-\lambda_0)\psi_0(t)+\lambda_0\phi_0\psi_1(t),\\
\frac{d \psi_{1}}{d t}(t)&=(\mu_1-\lambda_1)\psi_1(t)+\lambda_1\phi_1\psi_0(t),
\end{split}
\end{equation*}
with initial conditions $\psi_0(0)=\psi_1(0)=1$. The unique solution of this Cauchy problem is given by \eqref{expZG}.
\end{proof}
\begin{theorem}\label{teor:expHQ}
Suppose that $\int_E f_i(y)F_i(dy)<+\infty$ for $i=0,1$. Then the processes
\begin{equation}\label{eq:Ztilde}
L_t:=\sum_{n=1}^{N_t(E)}f_{\eps_n}\bigl(Y_{\eps_n,n}\bigr)-\int_0^t\int_{E}f_{\eps(s)}(y)\lambda_{\eps(s)}F_{\eps(s)}(d y)d s, \ \ t\geq0
\end{equation}
and
\begin{equation}
\mathcal{E}_t(L)=\prod_{n=1}^{N_t(E)}\bigl(1+f_{\eps_n}\bigl(Y_{\eps_n,n}\bigr)\bigr)\exp\left(-\int_0^t\int_{E}f_{\eps(s)}(y)\lambda_{\eps(s)}F_{\eps(s)}(d y)d s\right), \ \ t\geq0
\end{equation}
are $\Fil$-martingales.
\end{theorem}
\begin{proof}
Observe that $L_t$ is a jump-telegraph process with $\mu_i=-\lambda_i\eta_i.$ Then, by Theorem \ref{thm:meanZG} we have
\begin{equation}\label{antZ}
\Exp_i[L_t]=\Exp_i\left[\sum_{n=1}^{N_t(E)}f_{\eps_n}(Y_{\eps_n,n})-\int_0^t\int_{E}f_{\eps(s)}(y)\lambda_{\eps(s)}F_{\eps(s)}(d y)d s\right]=0,\quad i=0,1.
\end{equation}
Let $0\le s < t$ be fixed. Let $i\in\{0,1\}$ be the value of $\eps(\cdot)$ at time $s$ and let $k\in\Na$ be the value of $N_s(E)$ at time $s.$ By the strong Markov property, we have the following conditional identities in distribution
\begin{equation}\label{eq:markovQ}
\begin{aligned}
\bigl.\eps(s+u)\,\bigr|_{\{\eps(s)=i\}}&\overset{d}{=}\bigl.\tilde{\eps}(u)\,\bigr|_{\{\tilde{\eps}(0)=i\}}, \\
\bigl.\tau_{n+k}\,\bigr|_{\{\eps(s)=i\}}&\overset{d}{=}\tilde{\tau}_{n}\,\bigr|_{\{\tilde{\eps}(0)=i\}},
\end{aligned}
\quad
\begin{aligned}
\bigl.N_{s+u}(E)\,\bigr|_{\{\eps(s)=i\}}&\overset{d}{=}N_s(E)+\bigl.\widetilde{N}_u(E)\,\bigr|_{\{\tilde{\eps}(0)=i\}}, \\
\bigl.Y_{\eps_{n+k},n+k}\,\bigr|_{\{\eps(s)=i\}}&\overset{d}{=}\bigl.Y_{\tilde{\eps}_n,n}\,\bigr|_{\{\tilde{\eps}(0)=i\}},
\end{aligned}
\quad
\begin{aligned}
&u\geq0,\\
&n\geq0,
\end{aligned}
\end{equation}
where $\tilde{\eps}$, $\widetilde{N}$, $\{\tilde{\tau}_k\}$ and $\{Y_{\tilde{\eps}_k,k}\}$
are copies of the processes $\eps$, $N$, $\{\tau_k\}$ and $\{Y_{\eps_k,k}\}$, respectively, independent of $\mathcal{F}_s.$
Then, using \eqref{antZ} and \eqref{eq:markovQ}, we obtain
\begin{equation*}
\Exp[L_t-L_s\mid\mathcal{F}_s]
=\Exp_i\left[\sum_{n=1}^{\widetilde{N}_{t-s}(E)}f_{\tilde{\eps}_n}(Y_{\tilde{\eps}_n,n})
-\int_0^{t-s}\int_{E}f_{\tilde{\eps}(u)}(y)\lambda_{\tilde{\eps}(u)}F_{\tilde{\eps}(u)}(d y)d u\right]
= 0,
\end{equation*}
and the first part follows. Now, if we define the jump-telegraph process
\begin{equation*}
\hat{L}_t:=\sum_{n=1}^{N_t(E)}\log\bigl(1+f_{\eps_n}(Y_{\eps_n,n})\bigr)-\int_0^t\int_{E}f_{\eps(s)}(y)\lambda_{\eps(s)}F_{\eps(s)}(dy)ds,
\end{equation*}
then we have $\mathcal{E}_t(L)=e^{\hat{L}_t}$ and by Theorem \ref{thm:expZG} we find that $\Exp_i[e^{\hat{L}_t}]=1$, $i=0,1$. Using this and \eqref{eq:markovQ}, we obtain
\begin{equation*}
\Exp\bigl[e^{\hat{L}_t-\hat{L}_s}\mid\mathcal{F}_s\bigr]
=\Exp_i\Biggl[\exp\Biggl(\,\sum_{n=1}^{\widetilde{N}_{t-s}(E)}\log\bigl(1+f_{\tilde{\eps}_n}(Y_{\tilde{\eps}_n,n})\bigr)
-\int_0^{t-s}\int_{E}f_{\tilde{\eps}(u)}(y)\lambda_{\tilde{\eps}(u)}F_{\tilde{\eps}(u)}(dy)du\Biggr)\Biggr]= 1,
\end{equation*}
and the desired result follows.
\end{proof}

\begin{remark}\label{remZ:m2}
Let $Z^\varphi$ denote the Radon-Nikodym densities defined in the previous section. Then $Z_t^\varphi=\mathcal{E}_t(J)$ with
\begin{equation*}
J_t:=\sum_{n=1}^{N_t(E)}\bigl(\varphi_{\eps_n}(Y_{\eps_n,n})-1\bigr)
-\int_0^t\int_E\bigl(\varphi_{\eps(s)}(y)-1\bigr)\lambda_{\eps(s)}F_{\eps(s)}(dy)ds, \ \ t\in [0,T].
\end{equation*}
By Theorem \ref{teor:expHQ}, if $m=2,$ it is enough to have $\int_E\varphi_i(y)F_i(dy)<+\infty, \ i=0,1,$ to guarantee that $Z^\varphi$ is a $\Fil$-martingale. In particular, $\Exp[Z_T^\varphi]=1$ and we have $\varphi\in\Theta.$
\end{remark}

\begin{corollary}\label{lemrhodis}
The compensator $\rho(dy,dt)$ of the $E$-marked point process $\gamma(dy,dt)$ satisfies $\rho(dy,dt)=F_{\eps(t-)}(dy)\lambda_{\eps(t-)}\,dt,$ a.s.
\end{corollary}

\begin{proof}
Using Theorem \ref{teor:expHQ} with $f_i=\mathbf{1}_A,$ for $A\in\B(E),$ it follows that the process $M_t(A)$ defined as
\begin{align*}
 M_t(A):&=N_t(A)-\int_0^tF_{\eps(s-)}(A)\lambda_{\eps(s-)}ds\\
  &=\sum_{n=1}^{N_t(E)}\mathbf{1}_{\set{Y_{\eps_n,n}\in A}}-\int_0^t\int_A F_{\eps(s)}(dy)\,\lambda_{\eps(s)}\,ds\\
  &=\sum_{n=1}^{N_t(E)}\mathbf{1}_A(Y_{\eps_n,n})-\int_0^t\int_E\mathbf{1}_A(y)F_{\eps(s)}(dy)\,\lambda_{\eps(s)}\,ds, \ \ t\geq 0,
\end{align*}
is a $\Fil$-martingale. Then, for all bounded non-negative  $\Fil$-predictable process $\set{\phi_t}_{t\geq 0},$ the stochastic integral
\[
\int_0^t\phi_s\,dM_s(A)=\int_0^t\phi_s\,dN_s(A)-\int_0^t\phi_sF_{\eps(s)}(A)\lambda_{\eps(s)}ds, \ \ t\geq 0,
\]
is also a $\Fil$-martingale. By the Monotone Convergence Theorem, we have
\[
\Exp\left[\int_0^\infty\phi_s\,dN_s(A)\right]=\Exp\left[\int_0^\infty\phi_sF_{\eps(s-)}(A)\lambda_{\eps(s-)}\,ds\right].
\]
Hence, the counting process $N_t(A)$ is an inhomogeneous Poisson process with (Markov modulated) stochastic intensity  $F_{\eps(t-)}(A)\lambda_{\eps(t-)}.$ The desired result follows from Corollary T4 (Integration Theorem) in Br\'{e}maud \cite[Chapter VIII]{bremaud}.
\end{proof}



\section{Examples}
\subsection{Logarithmic utility}
We illustrate the main result first by considering logarithmic utility functions.
\begin{lemma}\label{betalog}
Let $U_1(t,x)=U_2(x)=\ln x.$ Then, for all $\varphi\in\tilde{\Theta}$ and $x>0$ we have $\beta^{x,\varphi}(t,y)=0,$ a.s. for $\rho$-a.e. $(t,y)\in [0,T]\times E.$
\end{lemma}
\begin{proof}
In this case, we have $I_1(t,y)=I_2(y)=1/y$ and $\calX^\varphi(y)=(T+1)/y,$ for $y\in(0,\infty).$ Then, $\Y^\varphi(x)=(T+1)/x,$ for $x\in(0,\infty),$
\begin{align}
  c_t^{x,\varphi} &=\frac{x}{(T+1)H_t^\varphi}, \ \ t\in [0,T],\label{clog}\\
  G^{x,\varphi}&=\frac{x}{(T+1)H_T^\varphi}.\notag
\end{align}
Hence, $M_t^{x,\varphi}=x$ for all $t\in [0,T],$ and the desired result follows.
\end{proof}

\begin{theorem}\label{thmlog}
Let $x$ be fixed. Suppose  Assumption \ref{A1} holds true and that for each $i\in\calI,$ there exists $\bar{\pi}_i$ satisfying $1+\bar{\pi}_if_i(y)>0$ for all $y\in E$ and
\begin{equation}\label{pif2}
\mu_i-r_i+\lambda_i\int_E \frac{f_i(y)}{1+\bar\pi_i f_i(y)}\,F_i(dy)=0.
\end{equation}
Suppose further there exists $p>1$ such that
\[
\int_E \frac{1}{[1+\bar\pi_i f_i(y)]^{p}}\,F_i(dy)<+\infty, \ \ \forall i\in\calI.
\]
Let
\[
\hat{\pi}_t:=\bar\pi_{\eps(t-)} \ \mbox{ and } \ \ \hat{c}_t:=\frac{V_t^{x,\hat{\pi},0}}{T+1}, \ \ t\in [0,T].
\]
Then
\begin{itemize}
   \item[(a)] The portfolio-consumption pair $(\hat{\pi},\hat{c})$ is optimal for $U_1(t,x)=U_2(x)=\ln x,$

 \item[(b)] The optimal wealth process $V^{x,\hat{\pi},\hat{c}}$ satisfies
 \[
 V^{x,\hat{\pi},\hat{c}}_t=V_t^{x,\hat{\pi},0}-t\hat{c}_t=V_t^{x,\hat{\pi},0}\Bigl(1-\frac{t}{T+1}\Bigr), \ \ t\in [0,T],
 \]

 \item[(c)] The optimal value function satisfies
\[
\vartheta_i(x)=(T+1)\ln x-(T+1)\ln (T+1)-\Exp_i\left[\int_0^T\ln V_t^{1,\hat{\pi},0}\,dt+\ln V_T^{1,\hat{\pi},0}\right], \ \ x>0, \ \ i\in\calI.
\]
\end{itemize}
\end{theorem}
\begin{proof}
 Define $\hat\varphi_i(y):=1/[1+\bar{\pi}_if_i(y)].$ By (\ref{pif2}) and Remark \ref{martcondZ}, the process $\hat{\varphi}$ belongs to $\tilde{\Theta}.$ By Lemma \ref{betalog}, $\hat\varphi$ and $\hat\pi$ satisfy the assumptions of Theorem \ref{main}. Then the pair $(\hat\pi,c^{x,\hat{\varphi}})$ is optimal.

Using again (\ref{pif2}), we see that the differential of $(H_t^{\hat{\varphi}})^{-1}$ satisfies
\begin{align*}
  d\Bigl(\frac{1}{H^{\hat{\varphi}}_t}\Bigr)
  &=\frac{1}{H^{\hat{\varphi}}_{t-}}\Bigl\{\Bigl[r_{\eps(t)}
  +\lambda_{\eps(t-)}\int_E(\hat{\varphi}_{\eps(t-)}(y)-1)\,F_{\eps(t-)}(dy)\Bigr]\,dt+\int_E\Bigl(\frac{1}{\hat{\varphi}_{\eps(t-)}(y)}-1\Bigr)\,\gamma(dy,dt)\Bigr\}\\
  &=\frac{1}{H^{\hat{\varphi}}_{t-}}\left\{r_{\eps(t)}\,dt+\bar\pi_{\eps(t-)}\left[\bigl(\mu_{\eps(t)}-r_{\eps(t)}\bigr)\,dt+\int_{E} f_{\eps(t-)}(y)\,\gamma(dy,dt)\right]\right\}.
\end{align*}
Hence, the process $(H_t^{\hat{\varphi}})^{-1}$ is a modification of $V_t^{1,\hat{\pi},0}.$ In view of (\ref{clog}), we conclude  $\hat{c}=c^{x,\hat{\varphi}},$ and (a) follows. Assertion (b) follows from (\ref{xiV}) and (c) follows from Theorem \ref{main}, part (c).
\end{proof}
Finally, we consider the case of two regimes for the underlying Markov chain.
\begin{corollary}
Let now $m=2.$ Assume that for each $i=0,1$ there exists $\bar{\pi}_i$ satisfying $1+\bar{\pi}_if_i(y)>0$ for all $y\in E$ as well as condition (\ref{pif2}). Assume further
\[
\int_E \left[\frac{1}{1+\bar\pi_i f_i(y)}+\ln(1+\bar\pi_if_i(y))\right]\,F_i(dy)<+\infty, \ \ i=0,1.
\]
Then the optimal value functions $\vartheta_i(x),$ $i=0,1,$ satisfy
\begin{align*}
\vartheta_0(x)&=(T+1)\ln x-(T+1)\ln (T+1)\\
&\phantom{=}-\frac{1}{2\lambda}\left\{(\lambda_1\bar{d}_0+\lambda_0\bar{d}_1)\Bigl(T+\frac{T^2}{2}\Bigr)
+\frac{\lambda_0(\bar{d}_0-\bar{d}_1)}{2\lambda}\Bigl[T+\bigl(1-e^{-2\lambda T}\bigr)\Bigl(1+\frac{1}{2\lambda}\Bigr)\Bigr]\right\}
\end{align*}
and
\begin{align*}
\vartheta_1(x)&=(T+1)\ln x-(T+1)\ln (T+1)\\
&\phantom{=}-\frac{1}{2\lambda}\left\{(\lambda_1\bar{d}_0+\lambda_0\bar{d}_1)\Bigl(T+\frac{T^2}{2}\Bigr)
-\frac{\lambda_1(\bar{d}_0-\bar{d}_1)}{2\lambda}\Bigl[T+\bigl(1-e^{-2\lambda T}\bigr)\Bigl(1+\frac{1}{2\lambda}\Bigr)\Bigr]\right\}
\end{align*}
where
\begin{equation*}
2\lambda:=\lambda_0+\lambda_1,\quad \bar\eta_i:=\int_{E} \ln(1+\bar\pi_if_i(y))\,F_i(dy),\quad\bar{d}_i:=\bar{\pi}_i\mu_i+(1-\bar\pi_i)r_i+\lambda_i\bar{\eta}_i, \ \ i=0,1.
\end{equation*}
\end{corollary}

\begin{proof}
Observe that
\[
\ln V_t^{1,\hat{\pi},0}=\int_0^t\left[\bar{\pi}_{\eps(s)}\mu_{\eps(s)}+\left(1-\bar\pi_{\eps(s)}\right)r_{\eps(s)}\right]\,ds
+\sum_{n=1}^{N_t(E)}\ln(1+\bar{\pi}_{\eps_n}f_{\eps_n}(Y_{\eps_n,n})), \ \ t\in[0,T],
\]
is a jump-telegraph process with alternating tendencies $\bar{\pi}_i\mu_i+(1-\bar\pi_i)r_i,$ $i=0,1.$ The desired result follows noting that
\[
\Exp_i\left[\int_0^T\ln V_t^{1,\hat{\pi},0}\,dt\right]=\int_0^T\Exp_i\bigl[\ln V_t^{1,\hat{\pi},0}]\,dt
\]
and using Theorem \ref{thm:meanZG} and Theorem \ref{thmlog}, part (c).
\end{proof}

\begin{example}
To illustrate the above result, let us assume $\mu_1<0<\mu_0$, $f_0(y)=f_1(y)=y$, $r_0,r_1\geq0$ and $r_0<\mu_0$. For regime $i=0$ we fix a parameter $\eta_0>0$ and assume $Y_{0,n}$ is supported on the interval $(-1,0)$ with distribution
\[
F_0(dy)=\eta_0(1+y)^{\eta_0-1}\mathbf{1}_{(-1,0)}(y)\,dy.
\]
The expected value is given by $\Exp[Y_{0,n}]=-\frac{1}{1+\eta_0}<0$. Notice that  the random variable $-V_{0,n}:=-\ln(1+Y_{0,n})$ is exponentially distributed with density function $\eta_0e^{-\eta_0v}, \ v>0.$

For regime $i=1$ we fix another parameter $\eta_1>1$ and assume $Y_{1,n}$ is supported on the interval $(0,\infty)$ with distribution
\[
F_1(dy)=\eta_1(1+y)^{-(\eta_1+1)}\mathbf{1}_{(0,\infty)}(y)\,dy.
\]
The expected value is given by $\Exp[Y_{1,n}]=\frac{1}{\eta_1-1}>0.$ In this case $V_{1,n}:=\ln(1+Y_{1,n})$ is exponentially distributed with density function $\eta_1e^{-\eta_1v}, \ v>0.$

For each regime we consider the following portfolio constraints: for $i=0$ we restrict $\pi$ to the interval $(-\infty,1),$ that is, borrowing (or short-selling of the money account) is not allowed, and for $i=1$ we restrict $\pi$ to the interval $(0,\infty),$ that is, short-selling of the risky asset is not allowed.

In view of these constraints, we define
\begin{equation}
\begin{split}
g_0(\pi):&=\mu_0-r_0+\lambda_0\int_{E}\frac{y}{1+\pi y}F_0(dy) \\
&=\mu_0-r_0+\lambda_0\eta_0\int_{-1}^0\frac{y(1+y)^{\eta_0-1}}{1+\pi y}dy,\quad \pi\in (-\infty,1),
\end{split}
\end{equation}
and
\begin{equation}
\begin{split}
g_1(\pi):&=\mu_1-r_1+\lambda_1\int_{E}\frac{y}{1+\pi y}F_1(dy) \\
&=\mu_1-r_1+\lambda_1\eta_1\int_{0}^\infty\frac{y(1+y)^{-(\eta_1+1)}}{1+\pi y}dy,\quad \pi\in (0,\infty).
\end{split}
\end{equation}
Both maps $g_0$ and $g_1$ are strictly decreasing on their respective domains
\begin{align*}
g_0'(\pi)&=-\lambda_0\eta_0\int_{-1}^0\frac{y^2(1+y)^{\eta_0-1}}{(1+\pi y)^2}dy<0, \ \ \forall \pi\in (-\infty,1), \\
g_1'(\pi)&=-\lambda_1\eta_1\int_{0}^\infty\frac{y^2(1+y)^{-(\eta_1+1)}}{(1+\pi y)^2}dy<0, \ \ \forall \pi\in (0,\infty).
\end{align*}

\begin{figure}[hb]
\centering
\includegraphics[scale=0.75]{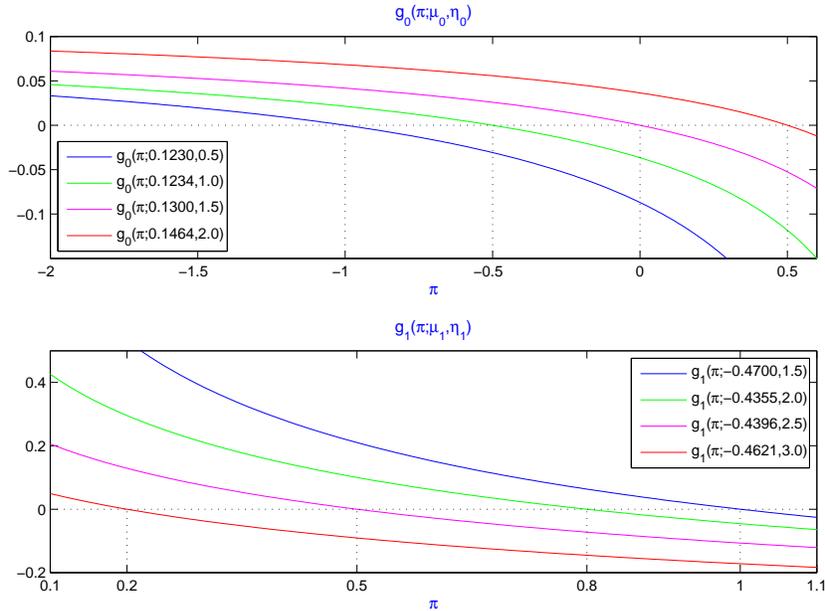}
\caption{\small Plots of $g_i(\pi)$ for  different values of $\mu_i$ and $\eta_i$.}
\end{figure}
Moreover,
\begin{equation*}
\lim_{\pi\to -\infty}g_0(\pi)=\mu_0-r_0>0, \quad \lim_{\pi\to +\infty}g_1(\pi)=\mu_1-r_1<0.
\end{equation*}
Hence, if there exists a pair $(\tilde\pi_0,\tilde\pi_1)$ satisfying $g_0(\tilde\pi_0)<0$ y $g_1(\tilde\pi_1)>0$, we can guarantee existence of solutions to equations $g_0(\pi)=0$ and $g_1(\pi)=0.$

We solved these equations numerically with $r_0=r_1=1\%,$ $\lambda_0=0.3$ and $\lambda_1=1.2.$ Figure 2 (above) shows plots of $g_i(\pi)$ for four different values of $\mu_i$ and $\eta_i,$ which have been chosen so that the optimal portfolio proportions are given by $\bar\pi_0=-1.0, -0.5, 0.0, 0.5$ and $\bar\pi_1=0.2, 0.5, 0.8, 1.0$ for each regime respectively. Figure 3 (below) plots the optimal portfolio proportion $\bar\pi_i$ as a function of $\mu_i,$ for different values of $\eta_i.$

\begin{figure}[hb]
\centering
\includegraphics[scale=0.75]{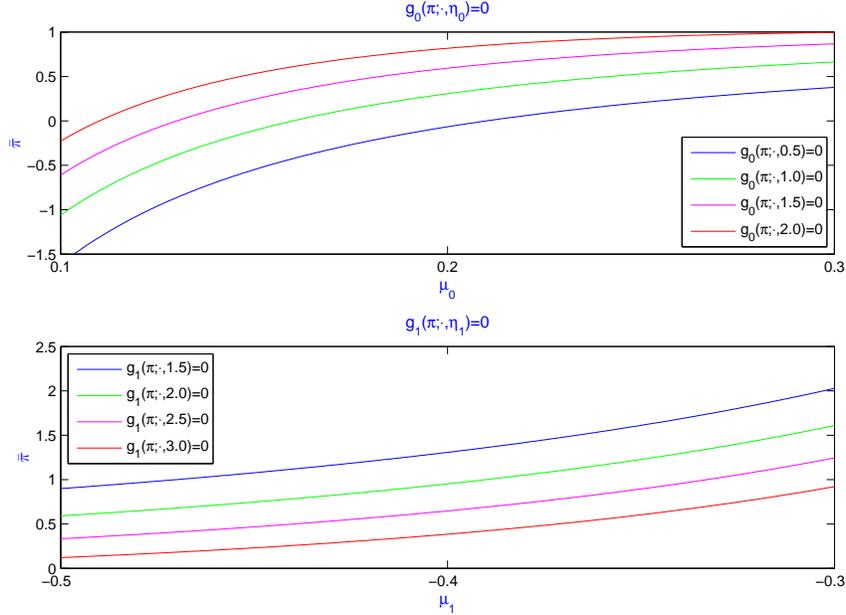}
\caption{\small Plots of optimal $\bar{\pi}_i$ as function of $\mu_i$, for different values of $\eta_i$.}
\end{figure}

\end{example}

\subsection{Fractional power utility}

Finally, we consider CRRA fractional power utility functions $U_1(t,x)=U_2(x)=\frac{x^\alpha}{\alpha}$ with $\alpha\in(0,1)$ fixed.


\begin{lemma}\label{betapower}
Let $\varphi\in\tilde{\Theta}$ be such that, for all $i\in \mathcal{I},$
\begin{equation}\label{varphilambda}
\int_E\left[\varphi_i(y)^{\frac{\alpha}{\alpha-1}}-\frac{\alpha}{\alpha-1}\varphi_i(y)\right]\,F_i(dy)=\frac{\alpha}{\alpha-1}\left(\frac{r_i-\lambda_i}{\lambda_i}\right).
\end{equation}
Then, for all $x>0,$ we have
\[
\beta^{x,\varphi}(t,y)= Y_{t-}^{x,\varphi}\bigl(\varphi_{\eps(t-)}(y)^{\frac{\alpha}{\alpha-1}}-1\bigr), \ \ \rho\mbox{-a.e.} \  (t,y)\in [0,T]\times E.
\]
\end{lemma}
\begin{proof}
Notice first that under condition (\ref{varphilambda}) the process $(H^\varphi)^{\frac{\alpha}{\alpha-1}}$ is a $\Fil$-martingale. Indeed, using integration formula for Marked point processes, we have
\begin{align*}
d\left[(H_t^\varphi)^{\frac{\alpha}{\alpha-1}}\right]
&=(H_{t-}^\varphi)^{\frac{\alpha}{\alpha-1}}\left\{\frac{\alpha}{\alpha-1}\Bigl[-r_{\eps(t)}-\lambda_{\eps(t)}\int_E(\varphi_{\eps(t)}(y)-1)\,F_{\eps(t)}(dy)\Bigr]\,dt\right.\\
&\phantom{AAAAAAAAA}+\left.\int_E(\varphi_{\eps(t-)}(y)^{\frac{\alpha}{\alpha-1}}-1)\,\gamma(dy,dt)\right\}\\
&=(H_{t-}^\varphi)^{\frac{\alpha}{\alpha-1}} \int_E\bigl(\varphi_{\eps(t-)}(y)^{\frac{\alpha}{1-\alpha}}-1\bigr)\,\tilde\gamma(dy,dt).
\end{align*}
Then
\begin{equation*}
\calX_i^\varphi(y)=y^{\frac{1}{\alpha-1}}\Exp_i\left[\int_0^T (H_t^\varphi)^{\frac{\alpha}{\alpha-1}}\,dt+(H_T^\varphi)^{\frac{\alpha}{\alpha-1}}\right]=y^{\frac{1}{\alpha-1}}(T+1).
\end{equation*}
It follows that
\[
\Y_i^{\varphi}(x)=\Bigl(\frac{x}{T+1}\Bigr)^{\alpha-1}, \quad x>0, \ \ i=0,1.
\]
and
\begin{align}
  c_t^{x,\varphi} &=\frac{x}{T+1}(H_t^\varphi)^{\frac{1}{\alpha-1}}, \ \ t\in [0,T],\label{cpower}\\
  G^{x,\varphi}&=\frac{x}{T+1}(H_T^\varphi)^{\frac{1}{\alpha-1}}.\notag
\end{align}
Hence
\[
Y_t^{x,\varphi}=\frac{x}{T+1}\Bigl[(H_t^\varphi)^{\frac{\alpha}{\alpha-1}}+\int_t^T(H_t^\varphi)^{\frac{\alpha}{\alpha-1}}\,ds\Bigr]=x (H_t^\varphi)^{\frac{\alpha}{\alpha-1}}\Bigl(1-\frac{t}{T+1}\Bigr)
\]
and
\[
M_t^{x,\varphi}=x (H_t^\varphi)^{\frac{\alpha}{\alpha-1}}\Bigl(1-\frac{t}{T+1}\Bigr)+\int_0^t\frac{x}{T+1}(H_s^\varphi)^{\frac{\alpha}{\alpha-1}}\,ds.
\]
Therefore, the differential of $M^{x,\varphi}$ satisfies
\begin{align*}
dM_t^{x,\varphi}&=x\left[\frac{-1}{T+1}(H_t^\varphi)^{\frac{\alpha}{\alpha-1}}\,dt+\Bigl(1-\frac{t}{T+1}\Bigr)\,d\left[(H_t^\varphi)^{\frac{\alpha}{\alpha-1}}\right]
+\frac{1}{T+1}(H_t^\varphi)^{\frac{\alpha}{\alpha-1}}\,dt\right]\\
&=x\Bigl(1-\frac{t}{T+1}\Bigr)(H_{t-}^\varphi)^{\frac{\alpha}{\alpha-1}} \int_E\bigl(\varphi_{\eps(t-)}(y)^{\frac{\alpha}{1-\alpha}}-1\bigr)\,\tilde\gamma(dy,dt)\\
&=Y_{t-}^{x,\varphi} \int_E\bigl(\varphi_{\eps(t-)}(y)^{\frac{\alpha}{1-\alpha}}-1\bigr)\,\tilde\gamma(dy,dt)
\end{align*}
and the desired result follows.
\end{proof}

\begin{theorem}\label{thmpower}
Let $x$ be fixed. Suppose  Assumption \ref{A1} holds  and that for each $i\in\calI,$ there exists $\bar{\pi}_i$ satisfying $1+\bar{\pi}_if_i(y)>0$ for all $y\in E,$
\begin{equation}\label{pif3}
\mu_i-r_i+\lambda_i\int_E \frac{f_i(y)}{[1+\bar\pi_i f_i(y)]^{1-\alpha}}\,F_i(dy)=0
\end{equation}
and
\begin{equation}\label{pif4}
\int_E\Bigl([1+\bar\pi_i f_i(y)]^{\alpha}-\frac{\alpha}{\alpha-1}[1+\bar\pi_i f_i(y)]^{\alpha-1}\Bigr)\,F_i(dy)=\frac{\alpha}{\alpha-1}\left(\frac{r_i-\lambda_i}{\lambda_i}\right).
\end{equation}
Let $\hat{\pi}_t:=\bar\pi_{\eps(t-)}$ and
\[
\hat{c}_t:=\frac{x}{T+1}\exp\left(\frac{-1}{\alpha}\int_0^t\int_E\left[1+\bar\pi_{\eps(s)}f_{\eps(s)}(y)\right]^\alpha F_{\eps(s)}(dy)\,\lambda_{\eps(s)}ds\right)
\prod_{n=1}^{N_t(E)}\left(1+\bar\pi_{\eps_n}f_{\eps_n}(Y_{\eps_n,n})\right)
\]
for $t\in [0,T].$ Then the portfolio-consumption pair $(\hat{\pi},\hat{c})$ is optimal for $U_1(t,x)=U_2(x)=\frac{x^\alpha}{\alpha}.$
\end{theorem}
\begin{proof}
Define $\hat\varphi_i(y):=1/[1+\bar{\pi}_if_i(y)]^{1-\alpha}.$ By (\ref{pif3}) and Remark \ref{martcondZ}, the process $\hat{\varphi}$ belongs to $\tilde{\Theta}.$ By Lemma \ref{betapower}, $\hat\varphi$ and $\hat\pi$ satisfy the assumptions of Theorem \ref{main}. Then the pair $(\hat\pi,c^{x,\hat{\varphi}})$ is optimal. Using (\ref{cpower}) and (\ref{pif4}),  we conclude that $\hat{c}=c^{x,\hat{\varphi}}$ and the desired result follows.
\end{proof}

\begin{remark}
Notice that both conditions (\ref{pif3}) and (\ref{pif4}) turn into a very specific constraint on the mean rate of return $\mu_i.$ Indeed, if there exists $\bar\pi_i$ satisfying (\ref{pif4}), then $\mu_i$ must satify condition (\ref{pif3}) for this value of $\bar\pi_i.$

On the other hand, if condition (\ref{pif3}) holds, it can be plugged into (\ref{pif4}) to obtain
\begin{equation*}
\int_E \frac{1}{[1+\bar\pi_i f_i(y)]^{1-\alpha}}\,F_i(dy)=\frac{1}{\lambda_i}[\bar\pi_i(1-\alpha)(\mu_i-r_i)+\alpha(\lambda_i-r_i)], \ \ i=0,1.
\end{equation*}
\end{remark}

To conclude, we have the following result for the two-regime case, which follows easily from (\ref{Vpi0}), Theorem \ref{thm:expZG} and Theorem \ref{betapower}.
\begin{corollary}
Let $m=2.$ Assume that consumption is not allowed i.e. $c_t=0$ for all $t\in [0,T],$ and for each $i=0,1$ there exists $\bar{\pi}_i$ satisfying $1+\bar{\pi}_if_i(y)>0$ for all $y\in E$ as well as conditions (\ref{pif3}) and (\ref{pif4}).
Then the optimal value functions $\vartheta_i(x),$ $i=0,1,$ satisfy
\[
\vartheta_0(x)=\frac{x^\alpha}{\alpha}e^{T(\bar\nu-\bar\lambda)}\Bigl[\cosh\bigl(T\sqrt{D}\,\bigr)
+\frac{\bar\mu-\bar\zeta+\lambda_0\phi_0}{\sqrt{D}}\sinh\bigl(T\sqrt{D}\,\bigr)\Bigr]
\]
and
\[
\vartheta_1(x)=\frac{x^\alpha}{\alpha}e^{T(\bar\nu-\bar\lambda)}\Bigl[\cosh\bigl(T\sqrt{D}\,\bigr)
-\frac{\bar\mu-\bar\zeta-\lambda_1\phi_1}{\sqrt{D}}\sinh\bigl(T\sqrt{D}\,\bigr)\Bigr]
\]
where
\begin{align*}
\bar\nu &:=\frac{\alpha}{2}\left[\bar\pi_0\mu_0+(1-\bar\pi_0)r_0+\bar\pi_1\mu_1+(1-\bar\pi_1)r_1\right]\\
\bar\mu &:=\frac{\alpha}{2}\left[\bar\pi_0\mu_0+(1-\bar\pi_0)r_0-\bar\pi_1\mu_1-(1-\bar\pi_1)r_1\right]\\
\bar\lambda &:=\frac{\lambda_0+\lambda_1}{2}\\
\bar\zeta &:=\frac{\lambda_0-\lambda_1}{2}\\
\phi_i&=\int_E[1+\bar\pi_i f_i(y)]^\alpha\,F_i(dy), \ \ i=0,1,
\end{align*}
and $D=(\bar\mu-\bar\zeta)^2+\lambda_0\lambda_1\phi_0\phi_1$.

\end{corollary}


\end{document}